\documentclass[a4paper,12pt]{article}
\usepackage{geometry,amssymb,graphicx,tensor,epigraph,amsmath,amsthm, xspace, listings, stmaryrd, tcolorbox, float, tikz-cd, mathtools, afterpage}
\usepackage{authblk}
\usepackage[T1]{fontenc}
\usepackage[utf8]{inputenc}
\usepackage[english]{babel}
\usepackage[textwidth=0.89in,textsize=scriptsize]{todonotes} 
\usepackage{hyperref}
\hypersetup{
	colorlinks=true, 
	linkcolor=blue,  
}

\theoremstyle{plain}
\newtheorem{prop}{Proposición}
\numberwithin{prop}{section}
\newtheorem*{prop*}{Proposición}

\newtheorem*{teo*}{Teorema}

\newtheorem*{coro*}{Corolario}

\newtheorem*{lema*}{Lema}

\theoremstyle{definition}

\newtheorem*{Def*}{Definición}
\newtheorem{Obs/Def}[prop]{Observación/Definición}

\theoremstyle{remark}

\newtheorem*{obs*}{Observación}
\newtheorem{Not/Obs}[prop]{Notación/Observación}
\newtheorem*{Not/Obs*}{Notación/Observación}

\numberwithin{ej_subseccion}{subsection}
\newtheorem*{ej*}{Ejemplo}
\newtheorem*{ej/def}{Ejemplo/Definición}

\newcommand{\keywords}[1]{\textbf{\textit{Keywords---}} #1}

\begin{document}
\title{Pool Value Replication (CPM) and Impermanent Loss Hedging}

\author[1,3]{Agustín Muñoz González}

\author[1,2]{\\ Juan I. Sequeira}

\author[]{\\ Ariel Dembling}

\affil[1]{Departamento de Matem\'aticas, Facultad de Ciencias Exactas y Naturales, Universidad de Buenos Aires, Buenos Aires, Argentina}
\affil[3]{POL Finance}

\affil[2]{IMAS-CONICET, Buenos Aires, Argentina}

\maketitle


\begin{abstract}
    This paper analytically characterizes impermanent loss for automated market makers (AMMs) in decentralized markets such as Uniswap or Balancer (CPMM). We derive a static replication formula for the pool's value using a combination of European calls and puts. Furthermore, we establish a result guaranteeing hedging coverage for all final prices within a predefined interval. These theoretical results motivate a numerical example where we illustrate the strangle strategy using real cryptocurrency options data from Deribit, one of the most liquid markets available.
\end{abstract}

\hspace{10pt}
\clearpage
\tableofcontents
\clearpage
\section*{Introduction}

A decentralized exchange (DEX) is a platform that enables cryptocurrency trading similarly to traditional exchanges, but unlike them, a DEX operates without the need for intermediaries. Such exchanges function via smart contracts, making them autonomous tools that anyone, anywhere in the world, can use without requiring permissions or authorizations from any central authority responsible for this service.

Historically, many different decentralized exchanges (DEXs) have been proposed using a variety of market-making or price discovery mechanisms, ranging from classical order book methods \cite{warren20170x} to more complex approaches involving specialized bonding curves \cite{hertzog2017bancor}. However, a simple yet surprisingly effective automated market maker seems to be the constant product market maker (CPMM), popularized by Uniswap \cite{zhang2018formal} and Balancer.

While order books are the dominant medium of electronic asset exchange in traditional finance \cite{wyart2008relation}, they are difficult to implement within a smart contract environment. The state size required by an order book to represent the set of pending orders is large and extremely costly within the smart contract setting, where users must pay for the storage and computational power utilized \cite{wood2017ethereum}. Furthermore, order matching logic is usually complex, as it often needs to support multiple different order types (such as iceberg, good-till-cancel, and stop-limit orders \cite{wyart2008relation}). To avoid significant on-chain execution costs (paid to the miners/validators of the smart contract by the agents executing trades), a variety of decentralized exchange designs use the underlying blockchain only for settlement, executing operations off-chain instead \cite{warren20170x}.

On the other hand, automated market makers (AMMs) have been extensively studied in algorithmic game theory, starting with Hanson's logarithmic market scoring rule (LMSR) \cite{hanson2003combinatorial}, which is often used in practice as an AMM for prediction markets. Such AMMs are typically constructed by liquidity providers first depositing assets in a fixed ratio to specify an initial belief distribution about possible outcomes. Then, the AMM provides a scoring rule specifying the cost of changing the belief distribution from its current state to a new desired state. This scoring rule incentivizes traders to truthfully reveal their belief that the expected value of adjusting the distribution is positive. Since the exchange state depends solely on the total deposited quantities, corresponding storage requirements are substantially lower than those of traditional exchanges. Additionally, pricing a trade requires only a single function evaluation, in contrast to more complicated matching algorithms like those used in order books.

A constant product market (CPM) is a mechanism to trade pairs of assets. A typical CPM is a pool containing certain amounts of both assets, along with a rule specifying how much of one asset will be exchanged for the other. Arbitrage with external spot markets ensures that the ratio of assets in the pool remains close to the prevailing exchange rate.

Liquidity providers supply the asset pair and can withdraw their liquidity at any time, receiving a proportion according to the new asset ratio, as well as a share of transaction fees. It can be shown that the payoff to liquidity providers is proportional to the square root of the price ratio between the two assets \cite{angeris2019analysis}. This payoff can be precisely replicated using a static combination of futures and options \cite{clark2020replicating}.

Liquidity providers are exposed to impermanent loss, which is only realized when liquidity is withdrawn from the pool. This loss is typically computed as the difference between the value of the tokens provided in the liquidity pool and the value of simply holding the tokens statically upon entering the pool. Since traders always exchange less valuable tokens for more valuable ones, liquidity providers consistently suffer impermanent losses (IL) that can be significant.

In this paper, we propose a static hedging strategy for liquidity providers using standard European options to eliminate the impact of impermanent loss. Firstly, we demonstrate that liquidity providers have an exposure equivalent to holding both long and short positions in various call and put options. We then derive a result ensuring impermanent loss protection within a given price interval $[P_i, P_s]$, provided certain inequalities relating to the quantities and costs of purchased puts and calls, total invested capital, and returns obtained from the pool are satisfied, considering the classic Long Strangle strategy.

In the results section, we illustrate this strategy through an application example of a hedging strategy, considering the duration we decide to remain in the pool, the initial capital, and the available options at the moment of implementing this project, utilizing real Ethereum option data from the Deribit options market.

\addcontentsline{toc}{section}{Introduction}
\section{Value of Constant Product Markets}

\subsection*{Constant Product Market}

A constant product market (CPM) is a market for exchanging tokens $X$ for tokens $Y$ (and vice versa). This market has reserves $x > 0$ and $y > 0$, a constant product $k = xy$, and a percentage fee $(1-\gamma)$. A transaction in this market, exchanging $\Delta_y > 0$ tokens $Y$ for $\Delta_x > 0$ tokens $X$, must satisfy:
\begin{equation}\label{ec: prod const}
(x - \Delta_x)(y + \gamma \Delta_y) = k.
\end{equation}

The reserves are updated as follows:
$y \to y + \Delta_y$, $x \to x - \Delta_x$, and $k \to (x - \Delta_x)(y + \Delta_y)$.

The term ``constant product market'' arises from the fact that when the fee is zero (i.e., $\gamma = 1$), any exchange $\Delta_y \leftrightarrow \Delta_x$ must alter the reserves such that the product $xy$ remains constant and equal to $k$.

\subsection*{Spot Market}

A spot market is a mechanism that exchanges $\Delta_x$ units of $X$ for $m_p^t \Delta_y$ units of $Y$ at time $t$. An infinitely elastic spot market is one where $m_p^t$ does not depend on $\Delta_y$.

\subsection*{The Value of a Constant Product Market}

Let $x^t$, $y^t$, and $m_p^t \in \mathbb{R}_+$ be the reserves of token $X$, token $Y$, and the relative market price of token $X$ in terms of token $Y$, respectively, at each time $t = 0, \dots, T$. On one hand, under the no-arbitrage hypothesis, the spot market price ($m_p^t$) coincides with the price implied by the constant product market ($y^t/x^t$), that is, $m_p^t = y^t/x^t$. On the other hand, if fees are not considered, the definition of a constant product market implies $x^t y^t \equiv k$ for all $t$. Combining these two statements yields:
\begin{equation}
    y^{t} = \sqrt{k m_{p}^{t}}.
\end{equation}

We can use this expression to compute the relative return between time $t-1$ and $t$:
\begin{equation*}
    \delta^{t} = \frac{m_p^{t} x^{t} + y^{t}}{m_p^{t-1} x^{t-1} + y^{t-1}} = \frac{y^t}{y^{t-1}} = \sqrt{\frac{m_p^t}{m_p^{t-1}}}.
\end{equation*}

Thus, the total relative gain is given by:
\begin{equation*}
    \delta = \prod_{t=1}^{T}\delta^{t} = \sqrt{\frac{m_p^{T}}{m_p^{0}}},
\end{equation*}
and the total portfolio value is:
\begin{equation}\label{valor del cpm}
    P_{V}^{T} = (m_p^{0} x^{0} + y^{0}) \delta = 2\sqrt{k m_p^{T}}.
\end{equation}

\section{Static Replication of a Payoff using Options and Bonds}

In the following section, we use \cite{carr2001towards}, which reviews the static replication theory using options originally developed by Ross \cite{ross1976options} and Breeden and Litzenberger \cite{breeden1978prices}.

Consider a situation in which investments are made at time $0$ and all payoffs are received at time $T$. In contrast to the standard intertemporal model, we assume that there are no trading opportunities except at times $0$ and $T$. We further assume that there exists a futures market for a risky asset (for example, a stock index) delivering at some date $T' \geq T$. We also assume that markets for European-style futures options exist\footnote{Note that futures options are typically American-style. However, when considering $T = T'$, the underlying futures converge to the spot at $T$, implying that European-style spot options effectively exist in this special case.} for all strikes. Although the assumption of a continuum of strikes is far from standard, it is essentially analogous to the conventional assumption of continuous trading. This assumption serves as a reasonable approximation in a setting where investors can trade frequently and when a large but finite number of strike options exist (for instance, futures options on the S\&P 500).

This market structure enables investors to replicate any smooth payoff function $f(P_T)$ of the futures' final price by taking a static position in options at time $0$\footnote{This observation was first noted by Breeden and Litzenberger \cite{breeden1978prices} and formally established by Green and Jarrow \cite{green1987spanning} and Nachman \cite{nachman1988spanning}}. In Appendix \ref{apendice}, we show how any payoff—i.e., any function $f:\mathbb{R}_+\to \mathbb{R}$ mapping the price $p$ of a risky asset to an amount $f(p)$—if required to be twice differentiable, can be expressed as:

\begin{equation}\label{replicador con opciones}
     f(P_T)=f(m)+f'(m)[(P_T-m)^+-(m-P_T)^+]+\\
     \int_{0}^{m}f''(K)(K-P_T)^+dK +\int_{m}^{\infty} f''(K)(P_T-K)^+dK,
\end{equation}
where $m\in\mathbb{R}^+$ and $f$ is continuous. The first term is interpreted as the payoff of a static position in $f(m)$ pure discount bonds (risk-free), each paying one dollar at time $T$. The second term can be interpreted as the payoff of $f'(m)$ calls with strike $m$ minus the payoff of $f'(m)$ puts, also with strike $m$. The third term represents a static position in $f''(K)dK$ puts at all strikes less than $m$. Similarly, the fourth term represents a static position in $f''(K)dK$ calls at all strikes greater than $m$.

In the absence of arbitrage, a decomposition similar to \eqref{replicador con opciones} must also hold between initial values. Let $V_0^f$ and $B_0$ be the initial values of the payoff and pure discount bond, respectively; $P_0(K)$ and $C_0(K)$ denote the initial prices of the put and call options with strike $K$, respectively. Thus, the present value of the aforementioned payoff is:

\begin{equation}\label{replicador con opciones inicila time}
     V_0^f = f(m)B_0 + f'(m)[C_0(m)-P_0(m)]+\\
     \int_{0}^{m} f''(K)P_0(K) dK + \int_{m}^{\infty} f''(K)C_0(K)dK.
\end{equation}

Consequently, the value of an arbitrary payoff can be determined solely from bond and option prices. Note that no assumption has been made regarding the stochastic process governing futures prices.

\section{Portfolio that Statically Replicates a Constant Product Market}

We use the previous result to statically replicate the value of a constant product market, $P_V^T = 2 \sqrt{k m^T}$, at a final price $m^T$.

Considering $f(-)=2\sqrt{k(-)}$, $P_T = m_p^T$ and $m = m_p^{0}$ in equation \eqref{replicador con opciones inicila time}, we obtain:
\begin{equation}
    P_V^{T}=f(m_p^{0})+f'(m_p^{0})[C(m_p^{0})-P(m_p^{0})]+\\
     \int_{0}^{m_p^{0}}f''(K)P(K)dK +\int_{m_p^{0}}^{\infty} f''(K)C(K)dK.
\end{equation}

\begin{itemize}
    \item Nominal value of the bond: $f(m^{0})= 2\sqrt{k m^{0}}$
    \item Notional value of options with strike $m_p^{0}$: $f'(m_p^{0})=\sqrt{\frac{k}{m^{0}}}$ 
    \item Notional value of options with strike $K$: $f''(K)=-\frac{1}{2}\sqrt{\frac{k}{K^3}}$
\end{itemize}

Since $(P_T - m)^+ - (m - P_T)^+ = P_T - m$, we can replicate the portfolio with a bond, futures, and options on the underlying asset, all maturing at $T$ with the specifications above, as follows:

\begin{equation}
    P_V^{T}=f(m_p^{0})+f'(m_p^{0})(m^T - m_p^{0})+\\
     \int_{0}^{m_p^{0}}f''(K)P(K)dK +\int_{m_p^{0}}^{\infty} f''(K)C(K)dK.
\end{equation}


\subsection{Example}

Consider a constant product market with the state:
    $$(t,x,y)=(0,200,10).$$

The initial value of the pool (in units of $Y$) is:
$$P_V^{0} = (m_p^{0} x^{0}+y^{0})=0.05\times 200 + 10 = 20.$$  

The replicating portfolio is:

\begin{itemize}
    \item Nominal value of the bond: $f(m_p^{0})= 2\sqrt{k m_p^{0}}= 2\sqrt{2000\times 0.05}=20$

    \item Notional value of futures $m_p^{0}$: $f'(m_p^{0})=\sqrt{\frac{k}{m_p^{0}}}= \sqrt{\frac{2000}{0.05}}=200$

    \item Notional value of options with strike $K$:  
    $$f''(K)=-\frac{1}{2}\sqrt{\frac{k}{K^3}} dK = -\frac{1}{2}\sqrt{\frac{2000}{K^3}} dK$$
\end{itemize}

If, for example, we have a discrete number of strikes $K=(0.125, 0.025, \dots , 0.1)$, the payoff at maturity is:

$$
    P_V^{T}=f(m_p^{0})+f'(m_p^{0})(m^T - m_p^{0})+
     \sum_{K\leq m_p^{0}}f''(K)P(K)\Delta K +\sum_{K>m_p^{0}} f''(K)C(K)\Delta K.
$$

Table \ref{tabla} shows the notional value for four strikes each for call and put options. Figure \ref{replicador} compares the total value of the portfolio and the replicating portfolio at time $T$.

\begin{figure}[!htt]
    	\centering
    	\includegraphics[width=10cm]{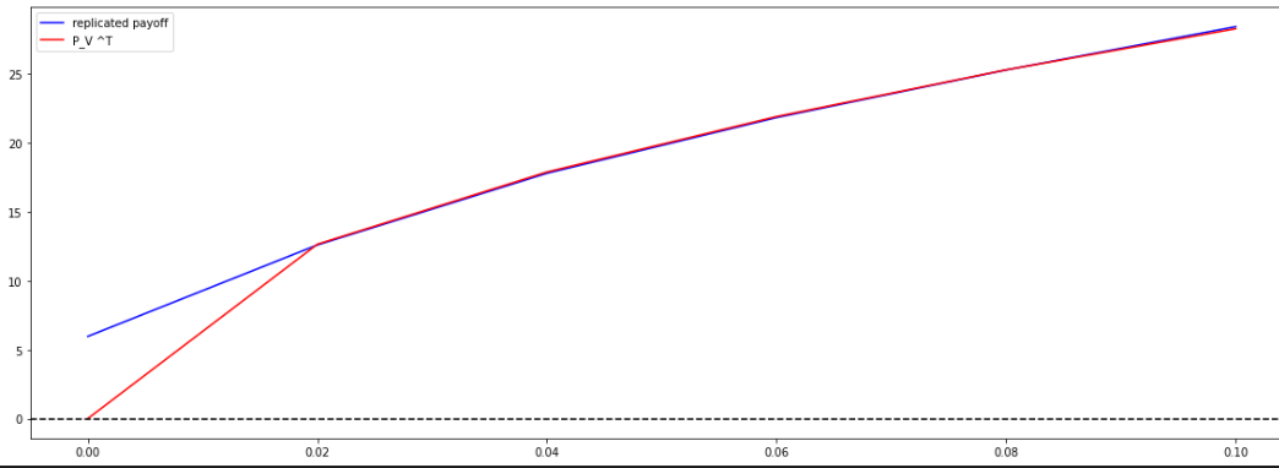}\caption{Value of the constant product market vs replicating portfolio}\label{replicador}
\end{figure}
    
\begin{figure}[!htt]
    	\centering
    	\includegraphics[width=10cm]{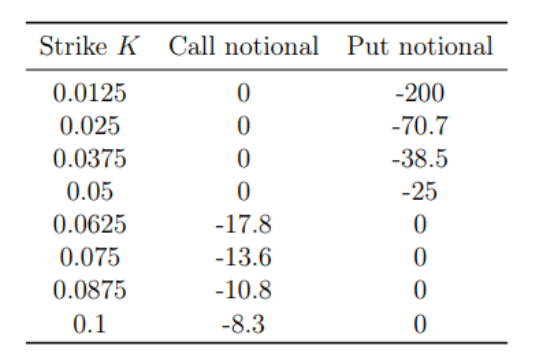}\caption{Notional value of options by strike}\label{tabla}
\end{figure}

\clearpage

\section{Hedging Against Impermanent Loss}
\subsection{Impermanent Loss}

Impermanent loss (IL) refers to the potential loss incurred when providing liquidity to a pool compared to holding the tokens statically outside the pool. Due to price fluctuations of token pairs, the impermanent loss materializes once the liquidity provider withdraws from the pool. We formally define impermanent loss as follows.

\begin{Def*}
For liquidity provision with initial deposits $x$ and $y$ of tokens $X$ and $Y$ at initial time $0$, the realized impermanent loss (IL) upon withdrawing liquidity at time $t$ is the capital loss compared to statically holding the token pair from the initial moment $0$. Specifically, impermanent loss (IL) is computed as:

$$IL = V_{pool} - V_{hold} = (y_t + x_t P_t) - (y_0 + x_0 P_t),$$
where $x_t$ and $y_t$ are amounts withdrawn at time $t$, and $P_t$ is the price of one unit of token $X$ denominated in units of token $Y$.
\end{Def*}

This definition aligns with industry practice, where liquidity providers regard this loss as the cost of repurchasing their initial liquidity upon exiting the pool.

\subsection{Formulas for IL}
Considering a constant product liquidity pool (CPM), we have:
\[
xy = k, \quad \frac{x}{y} = P.
\]

From these equations, we derive:
\[
x = \sqrt{\frac{k}{P}}, \quad y = \sqrt{k P}.
\]

We now express IL in terms of price:
\begin{align*}
    V_{Pool}(P) &= y + xP = 2\sqrt{kP_0}\sqrt{\frac{P}{P_0}} = V_{pool}(P_0)\sqrt{\frac{P}{P_0}}, \\
    V_{Hold}(P) &= y_0 + x_0P = \frac{V_{Hold}(P_0)}{2}\left(\frac{P}{P_0}+1\right).
\end{align*}

Given $V_{Pool}(P_0) = V_{Hold}(P_0)$, we have:
\[
IL(P) = V_{Hold}(P_0)\left(\sqrt{\frac{P}{P_0}}-\frac{1}{2}\left(\frac{P}{P_0}+1\right)\right).
\]

We then calculate the derivative of $IL(P)$ with respect to $P$:
\[
\frac{\partial}{\partial P} IL(P) = \frac{V_{Hold}(P_0)}{2P_0}\left(\sqrt{\frac{P_0}{P}} - 1\right).
\]

Thus, at any time $t$, we have:
\[
IL(P_t) = V_{Hold}(P_0)\left(\sqrt{\frac{P_t}{P_0}}-\frac{1}{2}\left(\frac{P_t}{P_0}+1\right)\right),
\]
and
\[
\frac{\partial}{\partial P} IL(P_t) = \frac{V_{Hold}(P_0)}{2P_0}\left(\sqrt{\frac{P_0}{P_t}} - 1\right).
\]

\subsubsection{Nonlinearity of IL}
Before presenting a specific hedging strategy, we highlight the inherent difficulty of hedging IL due to its multidirectional nonlinear nature.

Figure \ref{IL} illustrates impermanent loss for initial holdings of 50 ETH and 85,000 USDC, at a market price $P_{ETH} = \frac{85,000}{50} = 1700$, corresponding to an initial capital of $C = 50\times1700 + 85,000 = 170,000$.

\begin{figure}[ht]
    \centering
    \includegraphics[width=12cm]{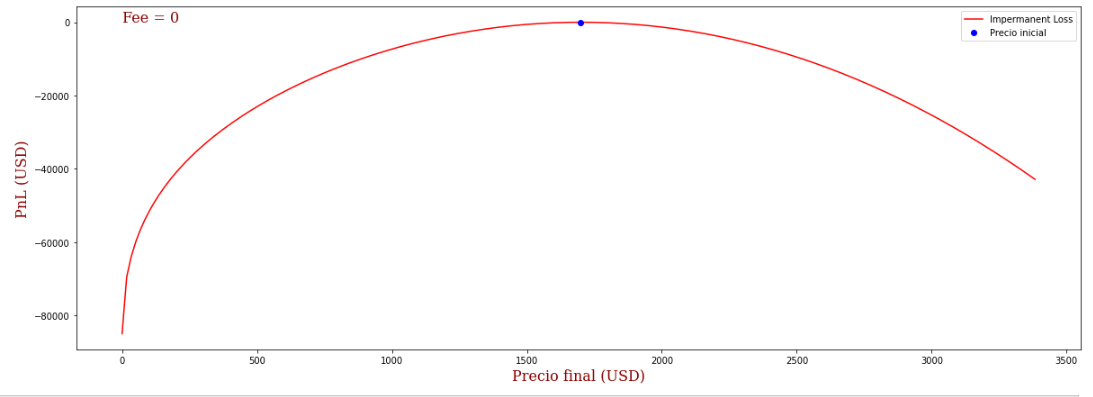}
    \caption{Impermanent Loss}\label{IL}
\end{figure}

Note that potential losses are asymmetric regarding price movements. A price decline impacts the provider more severely due to both the internal rebalancing of the pool and direct exposure to the declining asset.

\subsection{Static Hedging of IL with a Long Strangle Strategy}

Consider a liquidity provider entering a constant product pool by depositing $x_0$ USDC and $y_0$ Ethereum (ETH). Let $r_p$ be the monthly return rate paid by the pool. At time $T$, the impermanent loss from depositing tokens into the pool is:
\[
IL(P_T) = c\left(\sqrt{\frac{P_T}{P_0}}-\frac{1}{2}\left(\frac{P_T}{P_0}+1\right)\right),
\]
where $c := x_0 + y_0 P_0$ is the initial capital in USD and $P_0$ is the ETH price at entry.

Given the flexibility of options in replicating sufficiently smooth payoff functions, we utilize European options as our hedging strategy. Due to the nonlinearity of IL observed previously, we propose a Long Strangle strategy with expiration at time $T$. This strategy combines linear payoffs from calls and puts to offset IL nonlinearity. Specifically, we buy $q_c$ call options at strike $K_c$ and $q_p$ put options at strike $K_p$, with respective premiums $d_c$ and $d_p$.

The payoff of the option strategy at time $T$ is:
\[
q_c(P_T - K_c)^+ + q_p(K_p - P_T)^+ - D,
\]
where $D := q_c d_c + q_p d_p$ is the total cost of the options.

Thus, the combined pool and hedging strategy payoff at time $T$ is:
\[
f_T^{pool+str} = r_p c + payoff_{str} - D + IL(P_T).
\]

To avoid losses, we require $f_T^{pool+str}\geq 0$ over a limited interval $[P_i, P_s]$ containing $P_0$. The following proposition relates variables in this context:

\begin{prop}\label{prop:hed}
Suppose we know the Ethereum price $P_0$ at time $0$, the capital $c$ in dollars deposited in the pool at time $0$, and the pool return rate $r_p$ at time $T$. Then, to hedge impermanent losses at time $T$ within an interval $[P_i, P_s]\subset \mathbb{R}_{\geq 0}$ containing $P_0$, using a European option-based Long Strangle strategy, it suffices that the following inequalities are simultaneously satisfied:
\begin{align*}
    \frac{c}{2}\left(\frac{1}{\sqrt{P_iP_0}} - \frac{1}{P_0}\right)&\leq q_p,\\[5pt]
     D-\min\{IL(K_c),IL(K_p)\}&\leq r_p c,\\[5pt]
    -\frac{c}{2}\left(\frac{1}{\sqrt{P_s P_0}}-\frac{1}{P_0}\right)&\leq q_c,
\end{align*}
where $K_c, K_p$ denote the strikes; $d_c, d_p$ the premiums; $q_c, q_p$ the respective quantities of European call and put options expiring at $T$; and $D=q_c d_c+ q_p d_p$ is the strategy's total cost.
\end{prop}

\begin{proof}
To prove this statement, we analyze the following cases separately:
\begin{itemize}
    \item If $P\in [K_p,K_c]$, we have:
    \[
    f(P):=r_p c+q_c(P-K_c)^+ +q_p(K_p-P)^+ -D +IL(P) = r_p c - D + IL(P).
    \]
    Given a simple analysis of $IL$:
    \[
    IL(P)\geq \min\{IL(K_c),IL(K_p)\},
    \]
    thus,
    \[
    f(P)=r_p c - D + IL(P)\geq r_p c - D +\min\{IL(K_c),IL(K_p)\}.
    \]
    Therefore, if we impose:
    \[
    r_p c\geq D - \min\{IL(K_c),IL(K_p)\},
    \]
    it ensures $f(P)\geq 0$.
    
    \item If $P\in [P_i,K_p]$, we have:
    \[
    f(P)=r_p c+q_c(P-K_c)^+ +q_p(K_p-P)^+ -D +IL(P) = r_p c + q_p(K_p - P) - D + IL(P).
    \]
    Note that if $f(K_p)\geq 0$ and $f$ is decreasing within this interval ($f'<0$), then $f(P)\geq 0$. Given:
    \[
    f'(P)=IL'(P)-q_p =\frac{c}{2}\left(\frac{1}{\sqrt{PP_0}}-\frac{1}{P_0}\right)-q_p,
    \]
    and observing that:
    \[
    \frac{c}{2}\left(\frac{1}{\sqrt{PP_0}}-\frac{1}{P_0}\right)\leq \frac{c}{2}\left(\frac{1}{\sqrt{P_iP_0}}-\frac{1}{P_0}\right),
    \]
    it suffices to require:
    \[
    q_p\geq\frac{c}{2}\left(\frac{1}{\sqrt{P_iP_0}}-\frac{1}{P_0}\right).
    \]
    Moreover, observe that $f(K_p)\geq 0$ is already ensured by the first inequality.
    
    \item If $P\in [K_c, P_s]$, we have:
    \[
    f(P)=r_p c+q_c(P-K_c)^+ +q_p(K_p-P)^+ -D +IL(P) = r_p c + q_c(P - K_c) - D + IL(P).
    \]
    Note that if $f(K_c)\geq 0$ and $f$ is increasing in this interval ($f'>0$), then $f(P)\geq 0$. Given:
    \[
    f'(P)=IL'(P)+q_c =\frac{c}{2}\left(\frac{1}{\sqrt{PP_0}}-\frac{1}{P_0}\right)+q_c,
    \]
    and since:
    \[
    -\frac{c}{2}\left(\frac{1}{\sqrt{PP_0}}-\frac{1}{P_0}\right)< -\frac{c}{2}\left(\frac{1}{\sqrt{P_sP_0}}-\frac{1}{P_0}\right),
    \]
    it suffices to require:
    \[
    q_c\geq-\frac{c}{2}\left(\frac{1}{\sqrt{P_sP_0}}-\frac{1}{P_0}\right).
    \]
    Again, note that $f(K_c)\geq 0$ is satisfied by the initial condition.
\end{itemize}
Therefore, if all inequalities are satisfied simultaneously, we conclude that for any $P_T\in [P_i, P_s]$, it holds:
\[
f_T^{pool+str}=f(P_T)\geq 0.
\]
\end{proof}

\section{Hedging Against Impermanent Loss}
\subsection{Impermanent Loss}

Impermanent loss (IL) refers to the potential loss incurred when providing liquidity to a pool compared to holding the tokens statically outside the pool. Due to price fluctuations of token pairs, the impermanent loss materializes once the liquidity provider withdraws from the pool. We formally define impermanent loss as follows.

\begin{Def*}
For liquidity provision with initial deposits $x$ and $y$ of tokens $X$ and $Y$ at initial time $0$, the realized impermanent loss (IL) upon withdrawing liquidity at time $t$ is the capital loss compared to statically holding the token pair from the initial moment $0$. Specifically, impermanent loss (IL) is computed as:

$$IL = V_{pool} - V_{hold} = (y_t + x_t P_t) - (y_0 + x_0 P_t),$$
where $x_t$ and $y_t$ are amounts withdrawn at time $t$, and $P_t$ is the price of one unit of token $X$ denominated in units of token $Y$.
\end{Def*}

This definition aligns with industry practice, where liquidity providers regard this loss as the cost of repurchasing their initial liquidity upon exiting the pool.

\subsection{Formulas for IL}
Considering a constant product liquidity pool (CPM), we have:
\[
xy = k, \quad \frac{x}{y} = P.
\]

From these equations, we derive:
\[
x = \sqrt{\frac{k}{P}}, \quad y = \sqrt{k P}.
\]

We now express IL in terms of price:
\begin{align*}
    V_{Pool}(P) &= y + xP = 2\sqrt{kP_0}\sqrt{\frac{P}{P_0}} = V_{pool}(P_0)\sqrt{\frac{P}{P_0}}, \\
    V_{Hold}(P) &= y_0 + x_0P = \frac{V_{Hold}(P_0)}{2}\left(\frac{P}{P_0}+1\right).
\end{align*}

Given $V_{Pool}(P_0) = V_{Hold}(P_0)$, we have:
\[
IL(P) = V_{Hold}(P_0)\left(\sqrt{\frac{P}{P_0}}-\frac{1}{2}\left(\frac{P}{P_0}+1\right)\right).
\]

We then calculate the derivative of $IL(P)$ with respect to $P$:
\[
\frac{\partial}{\partial P} IL(P) = \frac{V_{Hold}(P_0)}{2P_0}\left(\sqrt{\frac{P_0}{P}} - 1\right).
\]

Thus, at any time $t$, we have:
\[
IL(P_t) = V_{Hold}(P_0)\left(\sqrt{\frac{P_t}{P_0}}-\frac{1}{2}\left(\frac{P_t}{P_0}+1\right)\right),
\]
and
\[
\frac{\partial}{\partial P} IL(P_t) = \frac{V_{Hold}(P_0)}{2P_0}\left(\sqrt{\frac{P_0}{P_t}} - 1\right).
\]

\subsubsection{Nonlinearity of IL}
Before presenting a specific hedging strategy, we highlight the inherent difficulty of hedging IL due to its multidirectional nonlinear nature.

Figure \ref{IL} illustrates impermanent loss for initial holdings of 50 ETH and 85,000 USDC, at a market price $P_{ETH} = \frac{85,000}{50} = 1700$, corresponding to an initial capital of $C = 50\times1700 + 85,000 = 170,000$.

\begin{figure}[ht]
    \centering
    \includegraphics[width=12cm]{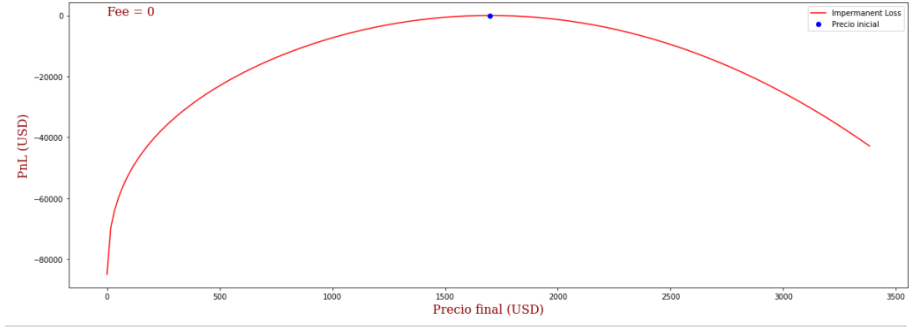}
    \caption{Impermanent Loss}\label{IL}
\end{figure}

Note that potential losses are asymmetric regarding price movements. A price decline impacts the provider more severely due to both the internal rebalancing of the pool and direct exposure to the declining asset.

\subsection{Static Hedging of IL with a Long Strangle Strategy}

Consider a liquidity provider entering a constant product pool by depositing $x_0$ USDC and $y_0$ Ethereum (ETH). Let $r_p$ be the monthly return rate paid by the pool. At time $T$, the impermanent loss from depositing tokens into the pool is:
\[
IL(P_T) = c\left(\sqrt{\frac{P_T}{P_0}}-\frac{1}{2}\left(\frac{P_T}{P_0}+1\right)\right),
\]
where $c := x_0 + y_0 P_0$ is the initial capital in USD and $P_0$ is the ETH price at entry.

Given the flexibility of options in replicating sufficiently smooth payoff functions, we utilize European options as our hedging strategy. Due to the nonlinearity of IL observed previously, we propose a Long Strangle strategy with expiration at time $T$. This strategy combines linear payoffs from calls and puts to offset IL nonlinearity. Specifically, we buy $q_c$ call options at strike $K_c$ and $q_p$ put options at strike $K_p$, with respective premiums $d_c$ and $d_p$.

The payoff of the option strategy at time $T$ is:
\[
q_c(P_T - K_c)^+ + q_p(K_p - P_T)^+ - D,
\]
where $D := q_c d_c + q_p d_p$ is the total cost of the options.

Thus, the combined pool and hedging strategy payoff at time $T$ is:
\[
f_T^{pool+str} = r_p c + payoff_{str} - D + IL(P_T).
\]

To avoid losses, we require $f_T^{pool+str}\geq 0$ over a limited interval $[P_i, P_s]$ containing $P_0$. The following proposition relates variables in this context:

\begin{prop}\label{prop:hed}
Given Ethereum's price $P_0$ at time $0$, initial capital $c$ deposited at time $0$, and pool return rate $r_p$ at time $T$, covering impermanent losses within the interval $[P_i, P_s]\subset \mathbb{R}_{\geq 0}$ containing $P_0$ using a European option Long Strangle strategy requires the simultaneous satisfaction of the following inequalities:
\begin{align*}
    \frac{c}{2}\left(\frac{1}{\sqrt{P_iP_0}} - \frac{1}{P_0}\right)&\leq q_p,\\
    D-\min\{IL(K_c),IL(K_p)\}&\leq r_p c,\\
    -\frac{c}{2}\left(\frac{1}{\sqrt{P_s P_0}}-\frac{1}{P_0}\right)&\leq q_c.
\end{align*}
Here, $K_c, K_p$ denote the strikes, $d_c, d_p$ the premiums, and $q_c, q_p$ the quantities of call and put options, respectively, all with expiration at $T$. The total cost of the strategy is $D=q_c d_c+ q_p d_p$.
\end{prop}

The proof involves analyzing price intervals separately to ensure the combined payoff remains non-negative.


\section{Appendix}\label{apendice}

\subsection{The Dirac Delta Function}
The Dirac delta or Dirac delta function is a distribution first introduced by British physicist Paul Dirac. As a distribution, it defines a functional in integral form over a certain space of functions. For more information, see \cite{balakrishnan2003all} and \cite{salah2015delta}.

To intuitively understand the Dirac delta function, consider a rectangle with one side along the $x$-axis, centered at $x = x_0$, such that the rectangle's area equals 1 (equivalent to a uniform probability distribution). Clearly, many such rectangles exist, as illustrated in Figure \ref{Dirac}. We can construct a Dirac delta function starting with a square of height and width 1. If we halve the width and double the height, the area remains constant. Repeating this process indefinitely, as width approaches zero, height approaches infinity, yet the area remains exactly 1. Any rectangle of unit area centered at $x_0$ can be expressed as:

\begin{equation*}
    \delta_{\varepsilon}(x - x_0)= \left\{ \begin{array}{lcc}
             0 &   if  & x < x_0 - \frac{\varepsilon}{2},  \\[6pt]
             \frac{1}{\varepsilon} &  if & x_0 - \frac{\varepsilon}{2} < x < x_0 + \frac{\varepsilon}{2}, \\[6pt]
             0 &  if  & x \geq x_0 + \frac{\varepsilon}{2}.
             \end{array}
   \right.
\end{equation*}

\begin{figure}[ht]\label{Dirac}
    \centering
    \includegraphics[width=10cm]{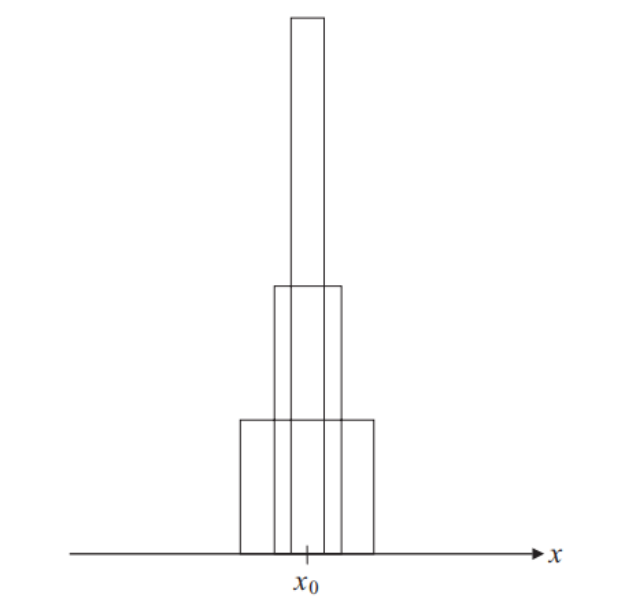}
    \caption{Geometric construction of the Dirac delta function}
\end{figure}

The Dirac delta function located at $x = x_0$ can be defined as the limit case when $\varepsilon$ tends to zero:

\begin{equation}
    \delta(x - x_0) = \lim_{\varepsilon \to 0} \delta_{\varepsilon}(x - x_0).
\end{equation}

A more general definition relies on fulfilling the following two properties:
\begin{align*}
    &\delta(x) = 0,\quad x \neq x_0,\\[5pt]
    &\int_{-\infty}^{+\infty}\delta(x)dx = 1.
\end{align*}

While the delta function has many properties, this work primarily uses the following:

\begin{itemize}
    \item \textbf{Sifting Property:} For any function $f(x)$ continuous at $x_0$,
    \begin{equation}
        \int_{-\infty}^{+\infty}f(x)\delta(x - x_0)dx = f(x_0),
    \end{equation}
    This property provides a measurement interpretation, indicating that the delta function "measures" the value of $f(x)$ at the point $x_0$.

    \item \textbf{Relationship with the Heaviside Step Function:}  
    The Heaviside step function is defined as:
    \begin{equation}
        H(x) = \left\{\begin{array}{lcc}
            0 & if & x < 0, \\[5pt]
            1 & if & x \geq 0.
        \end{array}\right.
    \end{equation}
    The delta function relates to the Heaviside step function by:
    \begin{equation}\label{deriv de H}
        \delta(x) = \frac{d}{dx}(H(x)).
    \end{equation}
    To verify this, one shows:
    \[
    \int_{-\infty}^{+\infty}H'(x)\phi(x)dx = \int_{-\infty}^{+\infty}\delta(x)\phi(x)dx,
    \]
    by integrating by parts.
\end{itemize}

\subsection{Expansion with Bonds and Options}

For any payoff function $f(F)$, by the sifting property of the Dirac delta, we have:
\begin{align}\label{ec: f con delta}
    f(F)&= \int_{0}^{+\infty} f(K)\delta(F - K)dK\\[5pt]
    &=\int_0^{\kappa}f(K)\delta(F - K)dK+\int_{\kappa}^{+\infty}f(K)\delta(F - K)dK,
\end{align}
for any non-negative $\kappa$. Integrating by parts each term, we use the following properties:

\begin{itemize}
    \item $\frac{d}{dx}(1(F < x)) = \delta(F - x)$,
    \item $\frac{d}{dx}(1(F \geq x)) = -\delta(x - F)$,
    \item $\frac{d}{dx}(-(F - x)^+) = 1(F \geq x)$,
    \item $\frac{d}{dx}((x - F)^+) = 1(F < x)$,
\end{itemize}
where
\begin{equation*}
    1(x \leq F)= \left\{\begin{array}{lcc}
             1 & if & x \leq F, \\[5pt]
             0 & if & x > F,
             \end{array}\right. \quad\text{and}\quad (F - x)^+ = \max\{F - x, 0\},
\end{equation*}
and similarly for other cases. The first two properties follow analogously to equation \eqref{deriv de H}, while the last two are straightforwardly verified.

Integrating each term of \eqref{ec: f con delta} by parts yields:
\begin{align*}
    f(F) &= f(K)1(F < K)\Big|_{0}^{\kappa} - \int_{0}^{\kappa} f'(K)1(F < K)dK \\[5pt]
    &- f(K)1(F \geq K)\Big|_{\kappa}^{\infty} + \int_{\kappa}^{\infty} f'(K)1(F \geq K)dK.
\end{align*}
Integrating again by parts, we obtain:
\begin{align*}
    f(F)&= f(\kappa)1(F < \kappa) - f'(K)(\kappa - F)^+\Big|_0^{\kappa} + \int_0^{\kappa} f''(K)(K - F)^+ dK\\[5pt]
    &+ f(\kappa)1(F \geq \kappa) - f'(K)(F - K)^+\Big|_{\kappa}^{\infty} + \int_{\kappa}^{\infty} f''(K)(F - K)^+ dK\\[5pt]
    &= f(\kappa) + f'(\kappa)\left[(F - \kappa)^+ - (\kappa - F)^+\right] \\[5pt]
    &+ \int_{0}^{\kappa} f''(K)(K - F)^+ dK + \int_{\kappa}^{\infty} f''(K)(F - K)^+ dK.
\end{align*}


\bibliographystyle{plain}
\bibliography{biblio}
\end{document}